%% file: entro.tex
\pdfoutput=1
\documentclass[10pt,a4paper,twocolumn,accepted=2018-03-20,superscriptaddress]{quantumarticle}
\usepackage[left=1.5cm,right=1.5cm,top=3cm,bottom=2cm]{geometry}
\usepackage[numbers,sort&compress]{natbib}
\usepackage[utf8]{inputenc}
\usepackage{amsmath}
\usepackage{amsfonts}
\usepackage{amssymb}
\usepackage{amsthm}
\usepackage{graphicx}
\usepackage{braket}
\usepackage[utf8]{inputenc}
\usepackage[english]{babel}
\usepackage[T1]{fontenc}
\usepackage{amsmath}
\usepackage{hyperref}
\usepackage{tikz}
\usepackage{lipsum}
\usepackage{sidecap}
\usepackage{wrapfig}

\title{Additivity of entropic uncertainty relations}
\author{René Schwonnek}
\affiliation{Institut für Theoretische Physik, Leibniz Universität Hannover, Germany}
\date{\today}

\newtheorem{lem}{Lemma}
\newtheorem{cor}{Corollary}
\newtheorem{thm}{Theorem}
\newtheorem{remark}{Remark}
\newtheorem{prop}{Proposition}

\newcommand{\tr}{\operatorname{tr}}
\newcommand{\px}{p^{X}}
\newcommand{\py}{p^{Y}}

\newcommand{\x}{\mathbf{x}}
\newcommand{\y}{\mathbf{y}}
\newcommand{\p}{\mathbf{p}}
\newcommand{\q}{\mathbf{q}}
\renewcommand{\v}{\mathbf{v}}
\newcommand{\w}{\mathbf{w}}

\begin{document}
\maketitle%
\begin{bf}
We consider the uncertainty between two pairs of local projective measurements performed on a multipartite system. We show that the optimal bound in any linear uncertainty relation, formulated in terms of the Shannon entropy, is additive. This directly implies, against naive intuition, that the minimal entropic uncertainty can always be realized by fully separable states. Hence, in contradiction to proposals by other authors, no entanglement witness can be constructed solely by comparing the attainable uncertainties of entangled and separable states. 
However, our result gives rise to a huge simplification for computing global uncertainty bounds as they now can be deduced from local ones. 

Furthermore, we provide the natural generalization of the Maassen and Uffink inequality for linear uncertainty relations with arbitrary positive coefficients.

\end{bf}
\vspace{0.8cm}
\section*{Introduction}
\input{intro2}

\section{Uncertainty in bipartitions}
\label{ucrbipart}
\input{ucrbipart}

\section{Linear uncertainty relations}
\label{eucr}
\input{eucr}

\section{Additivity, implications and applications}
\label{mainres}
\input{setting}

\section{Lower bounds from $(p,q)$-norms}
\label{pqnorms}
\input{pqnorm}

\section{Additivity of bounds from multiplicativity of $(p,q)$-norms}
\label{additivity}

\input{additivity}

 
\section*{Outlook and conclusion}
\input{outlook}

\section*{Acknowledgements}
R.S. acknowledges K. Abdelkhalek, O. Gühne, A. Costa, I. Siemon, and R.F. Werner for all the helpful discussions and comments. Furthermore, R.S. acknowledges Coco, Inken and Lars for a careful reading and correcting of this manuscript. R.S. also acknowledges financial support by the RTG 1991 and the SFB DQ-Mat, both founded by the DFG, and the project Q.com-q founded by the BMBF. Finally R.S. acknowledges the hospitality of the Centro de Ciencias de Benasque Pedro Pascual and Petronilla, granted at the beginning of this project. The publication of this article was funded by the Open Access Fund of the Leibniz Universität Hannover.

\bibliographystyle{unsrtnat}
\bibliography{entrolib2}

\end{document}

%% file: intro2.tex
\pdfoutput=1

Uncertainty and entanglement are doubtless two of the most prominent and drastic properties that set apart quantum physics from a classical view on the world. Their interplay contains a rich structure, which is neither sufficiently understood nor fully discovered. 
In this work, we reveal a new aspect of this structure: the additivity of entropic uncertainty relations. 

For product measurements in a multipartition, we show that the optimal bound $c_{ABC\ldots}$ in a linear uncertainty relation satisfies
\begin{align}
c_{ABC\cdots}=c_A+c_B+c_C+\ldots\quad,
\end{align}
where $c_A,c_B,c_C,\ldots$ are bounds that only depend on local measurements.  
This result implies that minimal uncertainty for product measurements can always be realized by uncorrelated states. Hence, we have an example for a task which is not improved by the use of entanglement.

We will quantify the uncertainty of a measurement by the Shannon entropy of its outcome distribution. For this case, the corresponding linear uncertainty bound $c_{ABC\ldots}$ gives the central estimate in many applications like: entropic  steering witnesses \cite{schneeloch,ana,alberto2,chinaSchur}, uncertainty relations with side-information \cite{berta}, some security proofs \cite{furrer} and many more. 

When speaking about uncertainty, we consider so called {\it preparation uncertainty relations} \cite{heisenberg,kennard,robertson,variances,amu,colespiani,review1,povm1}. From an operational point of view, a preparation uncertainty describes fundamental limitations, i.e. a tradeoff, on the certainty of predicting outcomes of several measurements that are performed on instances of the same state.  
This should not be confused \cite{myfriends} with its operational counterpart named measurement uncertainty\cite{busch,mur,renes,abbot,entromur}. A measurement uncertainty relation describes the ability of producing a measurement device which approximates several incompatible measurement devices in one shot. 



The calculations in this work focus on uncertainty relations in a bipartite setting. However, all results can easily be generalized to a multipartite setting by an iteration of statements on bipartitions. The basic measurement setting, which we consider for  bipartitions, is depicted in Fig.~\ref{fig:bipartsetting}.
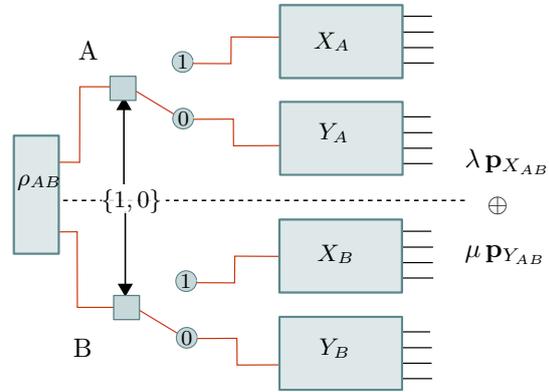
\begin{figure}[h]
\def\svgwidth{0.35\textwidth}
\input{settingbipart}
\caption{\label{fig:bipartsetting}\small Basic setting of product measurements on a bipartition: pairs of measurements $X_A,X_B$ or $Y_A,Y_B$ are applied to a joint state $\rho_{AB}$ at the respective sides of a bipartition. One bit of information is transmitted for communicating whether the $X$ or the $Y$ measurements are performed. The weights $(\lambda,\mu)$ denote the probabilities corresponding to this choice.}
\end{figure}
We consider a pair of measurements, $X_{AB}=X_A X_B$ and $Y_{AB}=Y_A Y_B$, to which we will refer as the global measurements of (tensor) product form. Each of those global measurements of product form is implemented by applying local measurements at the respective sides of a bipartition between parties denoted by $A$ and $B$. Hereby, the variables $X_A,X_B$ and $Y_A,Y_B$ will refer to those local measurements applied to the respective sides. 

We only consider projective measurements, but beside this we impose no further restrictions on the individual measurements. So the only property that measurements like $X_A$ and $X_B$ have to share is the common label '$X$', besides this, they could be non-commuting or even defined on Hilbert spaces with different dimensions.

The main result of this work is stated in Prop.\ref{prop1} in Sec. \ref{mainres}. In that section, we also collect some remarks on possible and impossible generalizations and the construction of entanglement witnesses. The proof of Prop.\ref{prop1} is placed at the end of this paper, as it relies on two basic theorems stated in Sec.\ref{pqnorms} and Sec.\ref{additivity}.

 Thm.\ref{thm1}, in Sec.\ref{pqnorms}, clarifies and expands the known connection between the logarithm of $(p,q)$-norms and entropic uncertainty relations. As a special case of this theorem we obtain Lem.\ref{lem1} which states the natural generalization of the well known Maassen and Uffink bound \cite{muff} to weighted uncertainty relations. Thm.\ref{thm2}, in Sec.\ref{additivity}, states that $(p,q)$-norms, in a certain parameter range, are multiplicative, which at the end leads to the desired statement on the additivity of uncertainty relations. 

Before stating the main result, we collect, in Sec.\ref{ucrbipart}, some general observations on the behavior of uncertainty relations for product measurements with respect to different classes of correlated states. Furthermore, in Sec.~\ref{eucr}, we will motivate and explain the explicit form of linear uncertainty relations used in this work.


%% file: settingbipart.tex
\pdfoutput=1

\begingroup%
  \makeatletter%
  \providecommand\color[2][]{%
    \errmessage{(Inkscape) Color is used for the text in Inkscape, but the package 'color.sty' is not loaded}%
    \renewcommand\color[2][]{}%
  }%
  \providecommand\transparent[1]{%
    \errmessage{(Inkscape) Transparency is used (non-zero) for the text in Inkscape, but the package 'transparent.sty' is not loaded}%
    \renewcommand\transparent[1]{}%
  }%
  \providecommand\rotatebox[2]{#2}%
  \ifx\svgwidth\undefined%
    \setlength{\unitlength}{235.1bp}%
    \ifx\svgscale\undefined%
      \relax%
    \else%
      \setlength{\unitlength}{\unitlength * \real{\svgscale}}%
    \fi%
  \else%
    \setlength{\unitlength}{\svgwidth}%
  \fi%
  \global\let\svgwidth\undefined%
  \global\let\svgscale\undefined%
  \makeatother%
  \begin{picture}(1,0.86540146)%
    \put(0,0){\includegraphics[width=\unitlength]{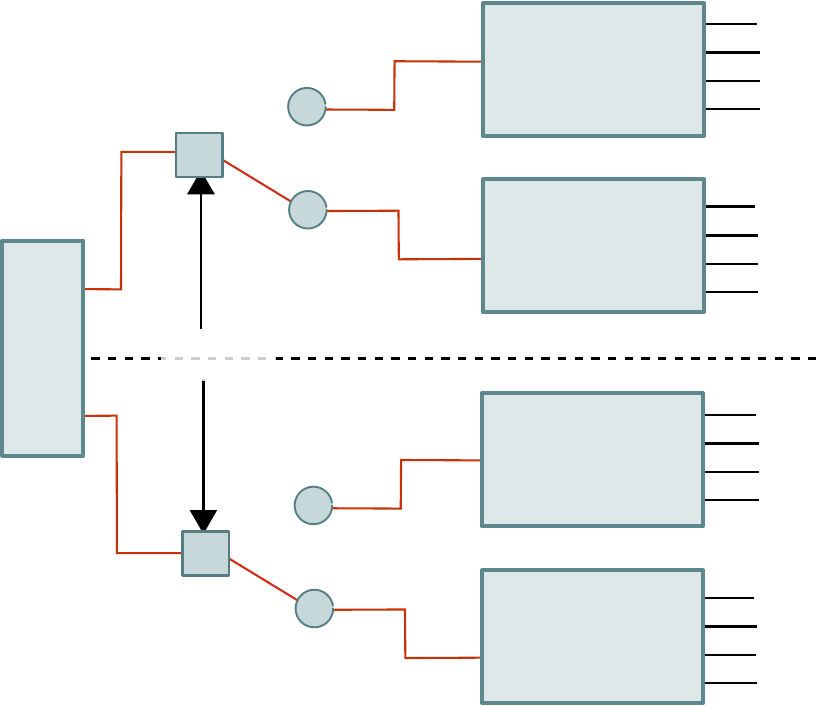}}%
    \put(0.0105934,0.45038045){\color[rgb]{0,0,0}\makebox(0,0)[lb]{\small{$\rho_{AB}$}}}%
    \put(0.66625045,0.75912322){\color[rgb]{0,0,0}\makebox(0,0)[lb]{\small{$X_A$}}}%
    \put(0.67672339,0.55572535){\color[rgb]{0,0,0}\makebox(0,0)[lb]{\small{$Y_A$}}}%
    \put(0.19641465,0.398){\color[rgb]{0,0,0}\makebox(0,0)[lb]{\small{$\{1,0\}$}}}%
    \put(0.37012529,0.23040606){\color[rgb]{0,0,0}\makebox(0,0)[lb]{\footnotesize{$1$}}}%
    \put(0.37245708,0.10397472){\color[rgb]{0,0,0}\makebox(0,0)[lb]{\footnotesize{$0$}}}%
    \put(0.67350572,0.28611338){\color[rgb]{0,0,0}\makebox(0,0)[lb]{\small{$X_B$}}}%
    \put(0.67790217,0.07542367){\color[rgb]{0,0,0}\makebox(0,0)[lb]{\small{$Y_B$}}}%
    \put(0.36415837,0.7188627){\color[rgb]{0,0,0}\makebox(0,0)[lb]{\footnotesize{$1$}}}%
    \put(0.36549016,0.592845494){\color[rgb]{0,0,0}\makebox(0,0)[lb]{\footnotesize{$0$}}}%
    \put(0.1451055,0.73881085){\color[rgb]{0,0,0}\makebox(0,0)[lb]{\smash{A}}}%
    \put(0.13441092,0.07866612){\color[rgb]{0,0,0}\makebox(0,0)[lb]{\smash{B}}}%
    \put(1,0.5){\color[rgb]{0,0,0}\makebox(0,0)[lb]{\smash{$\lambda \, \p_{X_{AB}}$}}}%
    \put(1.05,0.4){\color[rgb]{0,0,0}\makebox(0,0)[lb]{\smash{$\oplus$}}}%
    \put(1,0.3){\color[rgb]{0,0,0}\makebox(0,0)[lb]{\smash{$\mu \, \p_{Y_{AB}}$}}}%
  \end{picture}%
\endgroup%

%% file: ucrbipart.tex
\pdfoutput=1

All uncertainty relations considered is this paper are state-independent. 
In practice, finding a state-independent relation leads to the problem of jointly minimizing a tuple of given uncertainty measures, here the Shannon entropy of $X_{AB}$ and $Y_{AB}$, over all states. This minimum, or a lower bound on it, then gives the aforementioned trade-off, which then allows to formulate statements like: {\it "whenever the uncertainty of $X_{AB}$ is small, the uncertainty of $Y_{AB}$ has to be bigger than some state-independent constant" }.

Considering the measured state, $\rho_{AB}$, it is natural to distinguish between the three classes: uncorrelated, classically correlated and non-classical correlated. In regard of the uncertainty in a corresponding global measurement, states in these classes share some common features:

If the measured state is \textbf{uncorrelated}, i.e a product state $\rho_{AB}=\rho_A\otimes\rho_B$, the outcomes of the local measurements are uncorrelated as well.
Hence, the uncertainty of a global measurement is completely determined by the uncertainty of the local measurements on the respective local states $\rho_A$ and $\rho_B$. Moreover, in our case, the additivity of the Shannon entropy, tells us that the uncertainty of a global measurement is simply the sum of the uncertainty of the local ones. In the same way any trade-off on the global uncertainties can be deduced from local ones. 

If the measured state is \textbf{classically correlated}, i.e a convex combination of product states \cite{rfw89}, additivity of local uncertainties does not longer hold. More generally, whenever we consider a concave uncertainty measure \cite{gour}, like the Shannon entropy, the global uncertainty of a single global measurement is smaller than the sum of the local uncertainties. Intuitively this makes sense because a correlation allows to deduce information on the potential measurement outcomes of one side given a particular measurement outcome on the other.
However, a linear uncertainty relation for a pair of global measurements is not affected by this, i.e a trade-off will again be saturated by product states.
This is because the uncertainty relation between two measurements, restricted to some convex set of states, will always be attained on an extreme point of this set.  

However, if measurements are applied to an \textbf{entangled state}, more precisely to a state which shows EPR-steering \cite{epr,schrodinger,wiseman} with respect to the measurements $X_{AB}$ and $Y_{AB}$, it is in general not clear how a trade-off between global uncertainties relates to the corresponding trade-off between local ones. 
Just have in mind that steering implies the absence of any local state model, which is usually proven by showing that any such model would violate a local uncertainty relation. 

In principle one would expect to obtain smaller uncertainty bounds by also considering entangled states, and there are many entanglement witnesses known based on this idea (see also Rem.~\ref{rem3} in the following section).


%% file: eucr.tex
\pdfoutput=1

We note that there are many uncertainty measures, most prominently variances \cite{variances,kennard}. Variance, and similar constructed measures \cite{meandist,mur}, describe the deviation from a mean value, which clearly demands to assign a metric structure to the set of measurement outcomes. From a physicist's perspective this makes sense in many situations \cite{amu} but can also cause strange behaviours in situations where this metric structure has to be imposed artificially \cite{deutsch}. 
However, from the perspective of information theory, this seems to be an unnecessary dependency. Especially when uncertainties with respect to multipartitions are considered, it is not clear at all how such a metric should be constructed. Hence, it can be dropped and a quantity that only depends on probability distributions of measurement outcomes has to be used.  We will use the Shannon entropy. It fulfills the above requirement, does not change when the labeling of the measurement outcomes are permuted, and has a clear operational interpretation \cite{shannon,kolmogorov}. Remarkably, Claude Shannon himself used the term 'uncertainty' as an intuitive paraphrase for the quantity today known as 'entropy' \cite{shannon}. Historically, the decision to call the Shannon entropy an 'entropy' goes back to a suggestion John von Neumann gave to Shannon, when he was visiting Weyl in 1940 (there are, at least, three versions of this anecdote known \cite{neumannshannon2}, the most popular is \cite{neumannshannon}).   

Because we are not interested in assigning values to measurement outcomes, a measurement, say $X$, is sufficiently described by its POVM elements, $\{X_i\}$. So, given a state $\rho$, the probability of obtaining the $i$-th outcome is computed by $\tr(\rho X_i)$. The respective probability distribution of all outcomes is denoted by the vector $\p^x_\rho$. Within this notation the Shannon entropy of a $X$ measurement is given by $H(X|\rho):=-\sum_i \left(\p^x_\rho\right)_i \log \left(\p^x_\rho\right)_i$. As we restrict ourselves to non-degenerate projective measurements, all necessary information on a pair of measurements, $X$ and $Y$, is captured by a unitary $U$ that links the measurement basis. We will use the convention to write $U$ as transformation from the $\{X_i\}$ to the $\{Y_i\}$-basis, i.e. we will take $U$ such that $Y_i=U X_i U^\dagger$ holds.

Our basic objects of interest are optimal, state-independent and linear relations. This is, for fixed weights $\lambda,\mu\in\mathbb{R}^+$ we are interested in the best constant $c(\lambda,\mu)$ for which the linear inequality
\begin{align}
\lambda H(X|\rho)+ \mu H(Y|\rho) \geq c(\lambda,\mu) \label{wheightedbound} 
\end{align}
holds on \textit{all} states $\rho$.

Such a relation has two common interpretations: On one hand 
one can consider a guessing game, see also \cite{jed}. On the other, a relation like \eqref{wheightedbound} can be interpreted geometrically as in Fig.~\ref{fig:settingregion}.\vspace{0.2cm}

\noindent \textbf{Linear uncertainty: a guessing game}\\
For the moment, consider a player, called Eve, who plays against an opponent, called Alice. Dependent on a coin throw, in each round, Alice performs measurement $X_A$ or $Y_A$ on a local quantum state. Thereby the weights $\lambda$ and $\mu$ are the weights of the coin and the l.h.s. of  \eqref{wheightedbound} describes the total uncertainty Eve has on Alice's outcomes in each round. To be more precise, up to a $(\lambda,\mu)$-dependent constant, the l.h.s of \eqref{wheightedbound} equals the  Shannon entropy of the outcome distribution $\lambda \p^{X_{A}}_\rho \oplus \mu \p^{Y_{A}}_\rho$.

Eve's role in this game is to first choose a state $\rho$, observe the coin throw, wait for the measurements to be performed by Alice, and then ask binary questions to her opponent in order to get certainty on the outcomes.
Thereby, the Shannon entropy sum on the l.h.s of \eqref{wheightedbound} (with logarithm to the base 2) equals the expected amount questions Eve has to ask using an optimal strategy based on a fixed $\rho$.
Hence, the value $c(\lambda,\mu)$ denotes the minimal amount of expected questions, attainable by choosing an optimal $\rho$.

For a bipartite setting, Fig.~\ref{fig:bipartsetting}, a second player, say Bob, joins the game. Here, Eve will play the above game against Alice and Bob, simultaneously.
Thereby, Alice and Bob share a common coin, and,  therefore, apply measurements with the same labels ($X_{AB}$ or $Y_{AB}$). The obvious question that arises in this context is if Eve gets an advantage in this simultaneous game by using an entangled state or not.
 Prop.~\ref{prop1} in the next section answers the above question negatively, which is somehow unexpected as in principle the possible usage of non-classical correlations enlarges Eve's strategies. 
For example: Eve could have used a maximally entangled state, adjusted such that \textit{all} measurements Alice and Bob perform are maximally correlated. In this case the remaining uncertainty Eve has, would only be the uncertainty on the outcomes of one of the parties. However, the marginals of a maximally entangled state are maximally mixed. Hence, Eve still has a serious amount of uncertainty ($\log d)$, which turns out to be not small enough for beating a strategy based on minimizing the uncertainty of the local measurements individually. For the case of product-MUBs in prime square dimension \cite{ourentro}, it turns out that the minimal uncertainty realizable by a maximally entangled state actually equals the optimal bound.

\vspace{0.2cm}

\noindent \textbf{Linear uncertainty: the positive convex hull}\\
\begin{figure}[t]
\centering
\def\svgwidth{0.4\textwidth}
\includegraphics[width=0.45\textwidth]{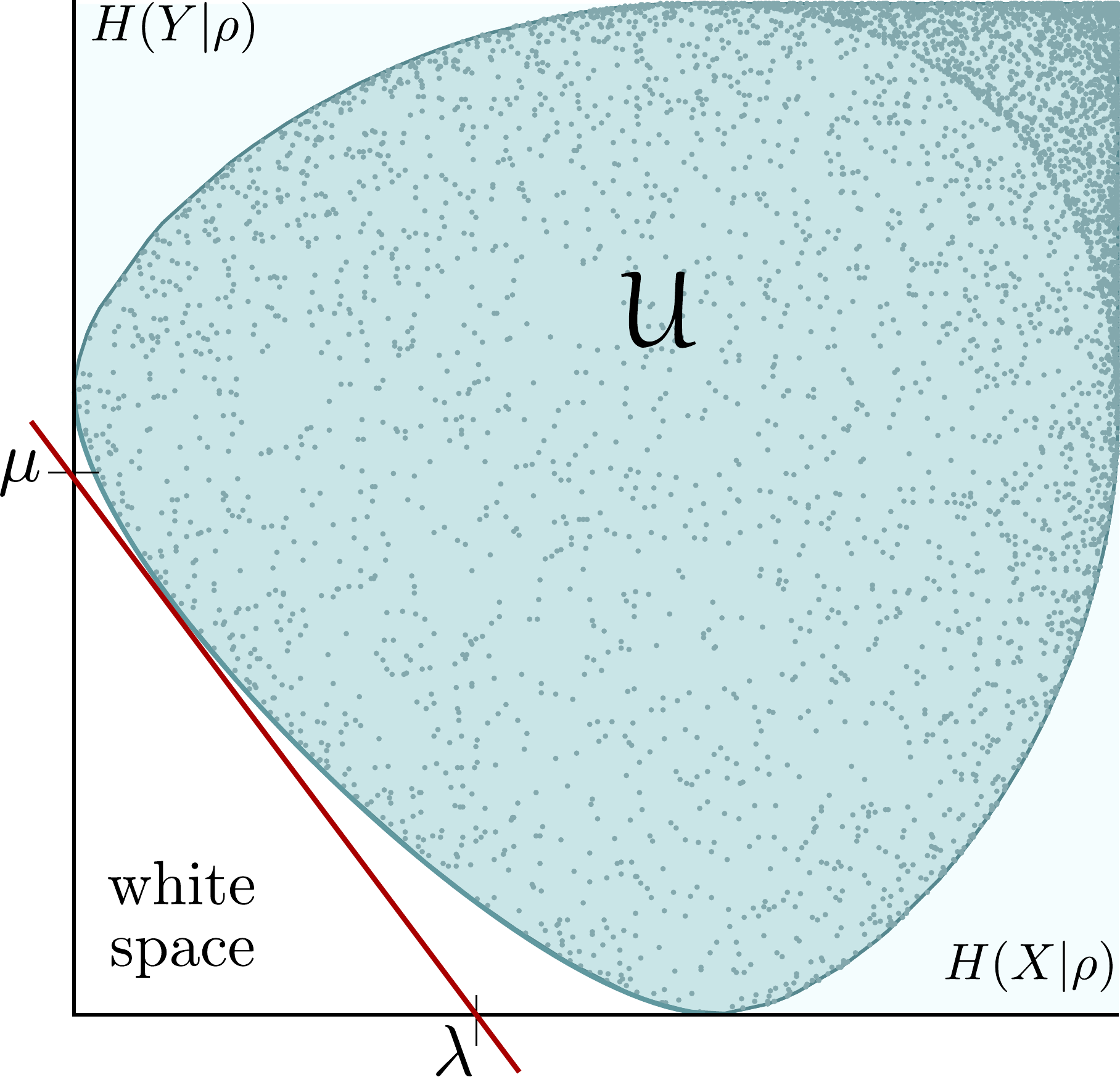}
\caption{Uncertainty set for measurements performed on a qubit. \label{fig:settingregion}
Any linear uncertainty relation, \eqref{wheightedbound}, with weights $(\lambda,\mu)$, gives the description of a tangent to the uncertainty set. All attainable pairs of entropies lie above this tangent.}
\end{figure} 
The second interpretation comes from considering the set of all attainable uncertainty pairs, the so called uncertainty set 
\begin{align}
\mathcal{U}=\left\{ \left(H(X|\rho),H(Y|\rho)\right) | \rho \textit{ is a quantum state} \right\}.
\end{align} 
In principle this set contains all information on the uncertainty trade-off between two measurements. More precisely, the white space in the lower-left corner of a diagram like Fig.~\ref{fig:settingregion} indicates that both uncertainties cannot be small simultaneously. In this context, a state-independent uncertainty gives a quantitative description of this white space.  
Unfortunately, it turns out that computing $\mathcal{U}$ can be very hard, because the whole state-space has to be considered.
Here a linear inequality, like \eqref{wheightedbound}, gives an outer approximation of this set. More precisely, if $c(\lambda,\mu)$ is the optimal constant in \eqref{wheightedbound}, this inequality describes a halfspace bounded from the lower-left by a tangent on $\mathcal{U}$. This tangent has the slope $\mu/\lambda$. The points on which this tangent touches the boundary of $\mathcal{U}$ corresponds to states which realize equality in \eqref{wheightedbound}. Those states are called minimal-uncertainty states. 
 Given all those tangents, i.e. $c(\lambda,\mu)$ for all positive $(\lambda,\mu)$, we can intersect all corresponding halfspaces and get a convex set which we call the {\it positive convex hull} of $\mathcal{U}$, denoted by $\overline{\mathcal{U}}$ in the following. Geometrically, the positive convex hull can be constructed by taking the convex hull of $\mathcal{U}$ and adding to it all points  that have bigger uncertainties then, at least, some point in $\mathcal{U}$. 
 
 If $\mathcal{U}$ is convex, like in the example above, $\overline{\mathcal{U}}$ contains the full information on the relevant parts of $\mathcal{U}$. If $\mathcal{U}$ is not convex, $\overline{\mathcal{U}}$ still gives a variety of state independent uncertainty relations, but there is still place for finding improvements, see \cite{ourentro}.

%% file: setting.tex
\pdfoutput=1

We are now able to state our main result
\begin{prop}[Additivity of linear uncertainty relations]\label{prop1}
Let $c_A(\lambda,\mu)$ and $c_B(\lambda,\mu)$ be state-independent lower bounds on the linear entropic uncertainty for local measurements $X_A,X_B$ and $Y_A,Y_B$, with weights $(\lambda,\mu)$. This means we have that
\begin{align}\label{prop4}
\lambda H(X_{A}|\rho_A)+ \mu  H(Y_{A}|\rho_A) \geq c_A(\lambda,\mu)\nonumber \\ 
\lambda H(X_{B}|\rho_B)+ \mu  H(Y_{B}|\rho_B) \geq c_B(\lambda,\mu)
\end{align}
holds on any state $\rho_A$ from $\mathcal{B}(\mathcal{H}_A)$ and $\rho_B$ from $\mathcal{B}(\mathcal{H}_B)$. 
Let $X_{AB}$ and $Y_{AB}$ be the joint global measurements that arise from locally performing $X_A,X_B$ and $Y_A,Y_B$ respectively. Then
\begin{align}
\lambda H(X_{AB}|\rho_{AB})+ \mu  H(Y_{AB}|\rho_{AB}) \geq c_{A}(\lambda,\mu)+c_{B}(\lambda,\mu) \label{eq:prop2}
\end{align}
holds for all states $\rho_{AB}$ from $\mathcal{B}(\mathcal{H}_{A}\otimes\mathcal{H}_B)$. Furthermore, if $c_A$ and $c_B$ are optimal bounds, then  
\begin{align}\label{eq:prop3}
c_{AB}(\lambda,\mu):=c_A(\lambda,\mu)+c_B(\lambda,\mu)
\end{align}
is the optimal bound in \eqref{eq:prop2}, i.e.  linear entropic uncertainty relations are additive.
\end{prop}

The proof of this proposition is placed at the end of Sec.~\ref{additivity}. We will proceed this section by collecting some remarks related to the above proposition: 
\begin{remark}[Product states]\label{rem1}\normalfont
Assume that 
$c_{A}(\lambda,\mu)$ and $c_{B}(\lambda,\mu)$ are optimal constants, and $\phi_A$ and $\phi_B$ are the states that saturate the corresponding uncertainty relations \eqref{prop4}. Then the product state $\phi_{AB}:=\phi_A\otimes\phi_B$ saturates \eqref{eq:prop2}, due to the additivity of the Shannon-entropy. 
However, this does not imply that all states that saturate \eqref{prop4} have to be product states. Examples for this, involving MUBs of product form, are provided in \cite{ourentro}. 
\end{remark}

\begin{remark}[Minkowski sums of uncertainty regions]\normalfont
Prop.~\ref{prop1} shows how the uncertainty set $\mathcal{U}_{AB}$, of the product measurement, relates to the uncertainty sets $\mathcal{U}_A$ and $\mathcal{U}_B$ of corresponding local measurements: For the case of an optimal $c_{AB}(\lambda,\mu)$, and fixed $(\lambda,\mu)$, equality in \eqref{eq:prop2} can always be realized by product states (see Rem.~\ref{rem1}). 
In an uncertainty diagram, like Fig.~\ref{fig:ucregionbipart}, those states correspond to points on the lower-left boundary of an uncertainty set, and, in general, they produce the finite extreme points of the positive convex hull of an uncertainty set.

\begin{figure}[h]
\includegraphics[width=0.48\textwidth]{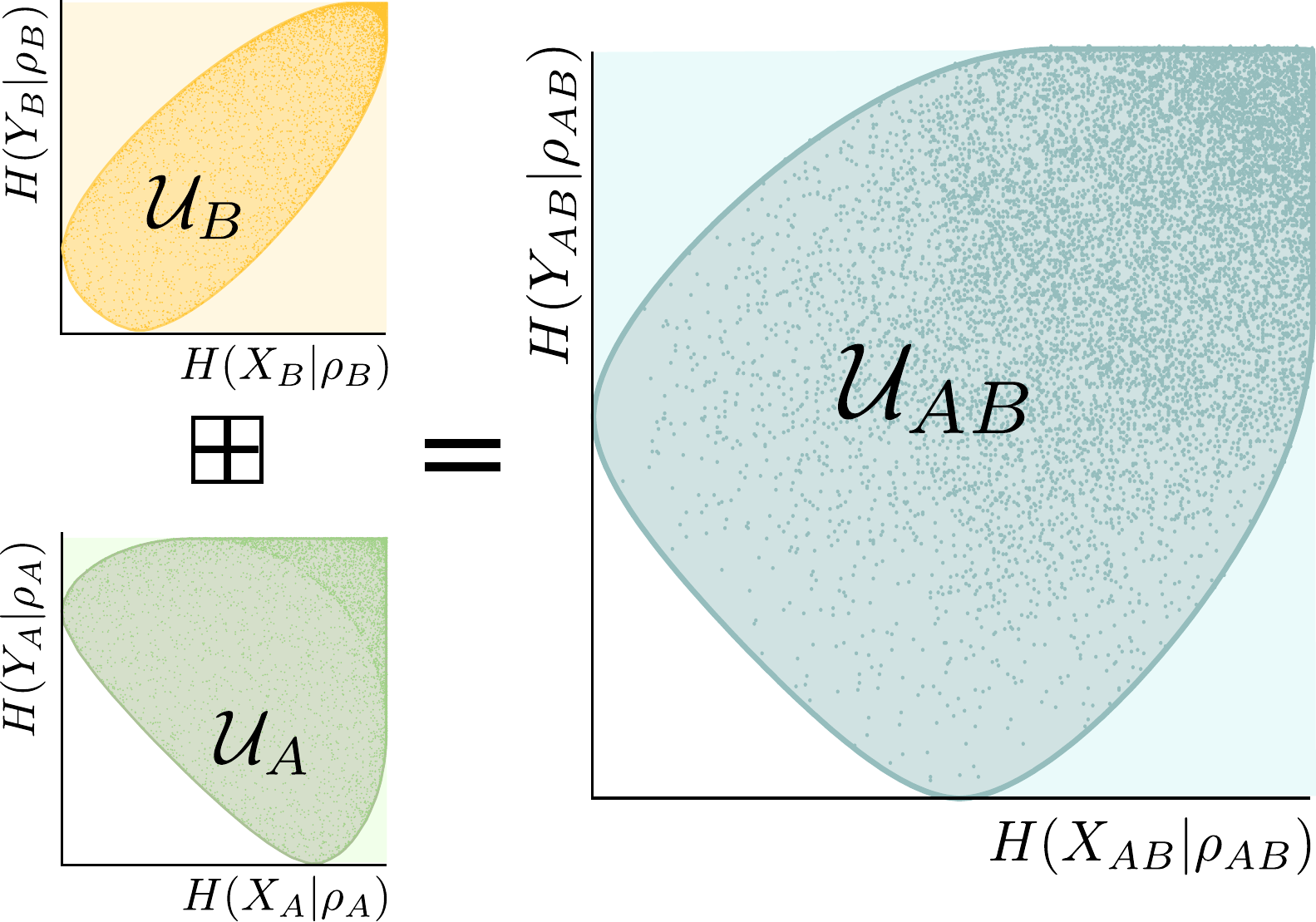}
\caption{\label{fig:ucregionbipart} Uncertainty sets of local measurements can be combined by the Minkowski sum: Uncertainty sets (green and yellow) for two pairs of local measurements on Qubits and the uncertainty set of the corresponding global measurements (blue).}
\end{figure}
For product states we have the additivity of the Shannon entropy, which gives
\begin{align}
\begin{pmatrix}
H(X_{AB}|\phi_A\otimes\phi_B)\\
H(Y_{AB}|\phi_A\otimes\phi_B)
\end{pmatrix} 
=
\begin{pmatrix}
H(X_{A}|\phi_A)\\
H(Y_{A}|\phi_A)
\end{pmatrix} 
+
\begin{pmatrix}
H(X_{B}|\phi_B)\\
H(Y_{B}|\phi_B)
\end{pmatrix} 
\end{align} 
This implies that we can get every extreme point of  $\overline{\mathcal{U}}_{AB}$ by taking the sum of two extreme points of $\overline{\mathcal{U}}_A$ and $\overline{\mathcal{U}}_B$. Due to convexity the same holds for all points in $\overline{\mathcal{U}}_{AB}$ and we can get this set as Minkowski sum \cite{minkowskisum}. 
\begin{align}
 \overline{\mathcal{U}}_{AB}=\overline{\mathcal{U}}_{AB} \boxplus\overline{\mathcal{U}}_B
\end{align} For convex uncertainty regions, arising from local measurements, this is depicted in Fig.~\ref{fig:ucregionbipart}. 
For this example, it is also true that $\mathcal{U}_{AB}$ itself is given as Minkowski sum of local uncertainty sets. However, we have to note, this behavior cannot be concluded from Prop.~\ref{prop1} alone. 
\end{remark}

\begin{remark}[Relation to existing entanglement witnesses]\label{rem3} \normalfont 
A well know method for constructing non-linear entanglement witnesses is based on computing the minimal value of a functional, like the sum of uncertainties \cite{hofmann,guhne2,spinsqueezing}, attainable on separable states. Given an unknown quantum state, the value of this functional is measured. If the measured value undergoes the limit set by separable states, the presence of entanglement is witnessed. 
For uncertainty relations based on the sum of general Schur concave functionals this method was proposed in \cite{chinaSchur}, including Shannon entropy, i.e.  the l.h.s. of \eqref{eq:prop2}, as central example. 

Our result Prop.~\ref{prop1} shows that this method will not work for Shannon entropies, because there is no entangled state that undergoes the limit set by separable states. We note that there is no mathematical contradiction between Prop.~\ref{prop1} and \cite{chinaSchur}. We only show that the set of examples for the method proposed in \cite{chinaSchur} is empty.

For uncertainty relations in terms of Shannon, Tsallis and Renyi entropies a similar procedure for constructing witnesses was  proposed by \cite{guhne, guhne2}. Here explicit examples for states, that can be witnessed to be entangled, were provided. Again, our proposition Prop.~\ref{prop1} is not in contradiction to this work because in \cite{guhne, guhne2} observables with a non-local degeneracy where considered. 

\end{remark}

Prop.~\ref{prop1} can easily be generalized to a multipartite setting, see Fig.~\ref{fig:settingmulti} : 
\begin{figure}[t]
\includegraphics[width=0.42\textwidth]{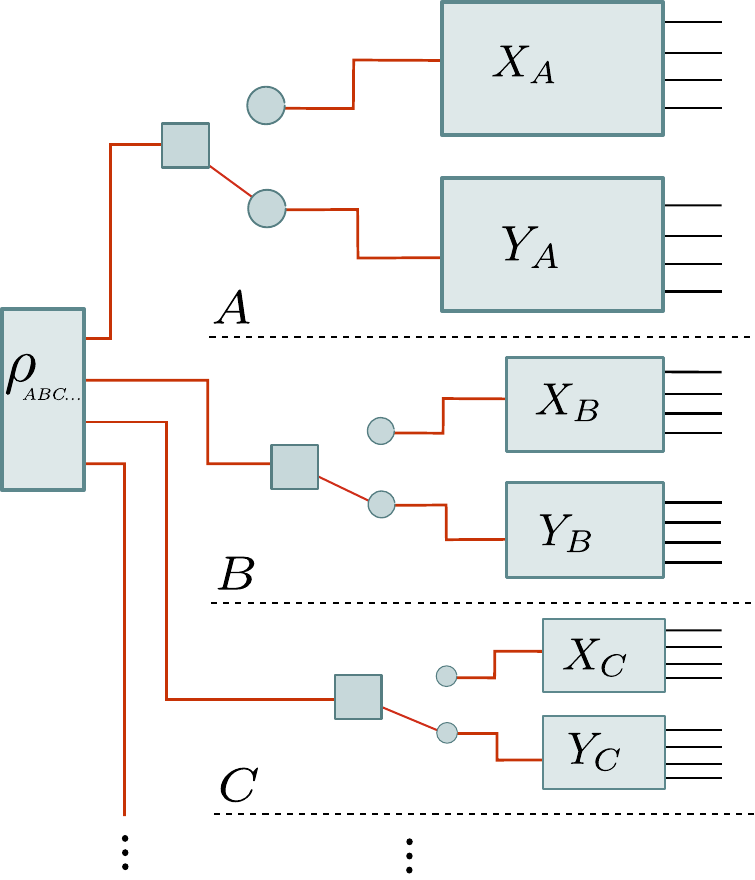}
\caption{\label{fig:settingmulti} Multiparite setting: Additivity of entropic uncertainty relations also holds if a pair of global product measurements for many local parties is considered.}
\end{figure}
\begin{cor}[Generalization to multipartite measurements]\normalfont
Assume parties $A_1\ldots A_n$ that locally perform measurements, $X_{A_1},\ldots,X_{A_n}$ or $Y_{A_1},\ldots,Y_{A_n}$, with weights 
$\vec\lambda=(\lambda_1,\ldots,\lambda_n )$.  
In analogy to \eqref{prop4}, let $c_{A_1}(\vec\lambda),\ldots,c_{A_n}(\vec\lambda)$ denote optimal local bounds and let $c_{A_1\cdots A_n}(\vec\lambda)$ be the optimal bound corresponding to product measurements $X_{A_1\ldots A_n}$ and $Y_{A_1\ldots A_n}$. We have
\begin{align}
c_{A_1\ldots A_n}(\vec\lambda)= \sum_{i=1}^n c_{A_i}(\vec\lambda)
\end{align}
This follows by iterating \eqref{eq:prop3}.
\end{cor}

\begin{remark}[Generalization to three measurements]\normalfont
The generalization of Prop.~\ref{prop1} to three measurements, say $X_{AB}$, $Y_{AB}$ and $Z_{AB}$, fails in general. The following  counterexample was provided by O. Gühne \cite{guhneexample}: 
For both parties we consider local measurements deduced from the three Pauli operators on a qubit and take all weights equal to one. In short hand notation we write $X_{AB}=\sigma_X\otimes \sigma_X$, $Y_{AB}=\sigma_Y\otimes \sigma_Y$, and $Z_{AB}=\sigma_Z\otimes \sigma_Z$. 
In this case, the minimal local uncertainty sum is attained on eigenstates  of the Pauli operators. If such a state is measured, the entropy for one of the measurements is zero and maximal for the others. Hence, the local uncertainty sum is always bigger than $2$ [\texttt{bit}]. Therefore we have
\begin{align}
&H\left(\sigma_X\otimes \sigma_X|\phi_A\otimes\phi_B \right)+&\nonumber\\&H\left(\sigma_Y\otimes \sigma_Y|\phi_A\otimes\phi_B \right)+&\nonumber\\&H\left(\sigma_Z\otimes \sigma_Z| \phi_A\otimes\phi_B\right)\geq 4
\end{align}
for all product states. 
In contrast to this a Bell state, say $\Psi^-$, will give the entropy of $1$[\texttt{bit}], for all above measurements. Hence we have,
\begin{align}
&H\left(\sigma_X\otimes \sigma_X|\Psi^-\right)+&\nonumber\\&H\left(\sigma_Y\otimes \sigma_Y|\Psi^-\right)+&\nonumber\\&H\left(\sigma_Z\otimes \sigma_Z| \Psi^-\right)=3\ngeq 4.
\end{align}
\end{remark}


%% file: pqnorm.tex
\pdfoutput=1

The quite standard technique for analyzing a linear uncertainty relation is to connect it to the $(p,q)$-\textit{norm} (see \eqref{pqnormdef} below) of the basis transformation $U$. Thereby, the majority of previous works in this field is concentrating only on handling the case of equal weights $\lambda=\mu=1$, which is connected to the $(p,q)$-norm for the case $1/p+1/q=1$. However, for the purpose of this work, i.e. for proving Prop.~\ref{prop1}, we have to extend this connection to arbitrary $(\lambda,\mu)$. We will do this by Thm.~\ref{thm1} on the next page.

A historically important example for the use of the connection between $(p,q)$-norms and entropic uncertainties, is provided by Bialynicki-Birula and Mycielski \cite{bbm}. They used Beckner's result \cite{beckner}, who computed the $(p,q)$-\textit{norm} of the Fourier-Transfromation, for proving the corresponding uncertainty relation, between position and momentum, conjectured by Hirschmann \cite{hirschmann}. Also Maassen and Uffink \cite{muff} took this way for proving their famous relation. Our result gives a direct generalization of this, meaning we will recover the Maassen and Uffink relation at the end of this section as special case of \eqref{generalmaassen}. Albeit, before stating our result, we will start this section by shortly reviewing the previously known way for connecting  $(p,q)$-norms with linear uncertainty relation, see also \cite{hanstalk,hans2} for further details:

The $(p,q)$-norm, i.e the $l^p\rightarrow l^q$ operator norm, of a basis transformation $U$ is given by
\begin{align}
\label{pqnormdef}
\Vert U \Vert_{q,p} :=\sup_{\phi\in\mathcal{H}} \frac{\Vert U \phi\Vert_q }{\Vert \phi \Vert_p } .
\end{align}
Here, the limit of $\Vert U \Vert_{q,p}$ for $(p,q)\rightarrow(2,2)$ goes to $1$. However, when $p$ and $q$ are fixed on the curve $1/p+1/q=1$, the leading order of $\Vert U \Vert_{q,p}$ around $(p,q)=(2,2)$ recovers the uncertainty relation \eqref{wheightedbound} in the case of equal weights $\lambda=\mu=1/2$, see \cite{hirschmann,bbm}. 

More precisely, taking the negative logarithm of \eqref{pqnormdef} gives 
\begin{align}\label{renyi1}
-\log\Vert U \Vert_{q,p}=\inf_{\phi\in\mathcal{H}} \; \log \Vert\phi\Vert_p - \log \Vert  U  \phi\Vert_q \; .
\end{align}
Here, we can identify the squared modulus of the components of $\phi$ as  probabilities of the $X$ and $Y$ measurement outcomes
\begin{align}
|(\phi)_i|^2&=\bra{\phi} X_i \ket{\phi}=(\p_\phi^X)_i \nonumber\\ |(U\phi)_i|^2&=\bra{\phi} Y_i \ket{\phi}\;=(\p_\phi^Y)_i 
\end{align}
and substitute 
\begin{align}\label{renyi2}
\Vert\phi\Vert_p = \left( \Vert  \p_\phi^X \Vert_{p/2} \right)^2 \text{\; and \;} \Vert U \phi\Vert_{q} =\left( \Vert \p_\phi^Y \Vert_{q/2} \right)^2 \;.
\end{align}
By this, \eqref{renyi1} gives a linear relation in terms of the $\alpha$-Renyi entropy \cite{renyi}, $H_\alpha(\p)=\frac{\alpha}{1-\alpha}\log(\Vert \p \Vert_\alpha )$. Here we get  
\begin{align}\label{renyi3}
\inf_{\phi\in\mathcal{H}} \frac{2-p}{p} H_{p/2}(X|\phi)- \frac{2-q}{q}  H_{q/2}(Y|\phi) = -\log\Vert U \Vert_{q,p}^2\; . 
\end{align}
If we evaluate this on the curve $1/p+1/q=1$, for $p\leq 2\leq q$, we can use 
\begin{align}
\frac{2-p}{p}=\frac{1}{p}-\frac{1}{q}=\frac{q-2}{q},
\end{align} 
 which can be employed to \eqref{renyi3}, in order to get
\begin{align}
\label{renyi4}
\inf_{\phi\in\mathcal{H}}  \; H_{p/2}(X|\phi)+H_{q/2}(Y|\phi) = \left(\frac{1}{q}-\frac{1}{p}\right)^{-1}\log\Vert U \Vert_{q,p}^2\; . 
\end{align}
Here, the limit  $(p,q)\rightarrow(2,2)$, in the l.h.s of \eqref{renyi4}, gives the limit from the Renyi to the Shannon entropy. This gives the l.h.s. of the uncertainty relation \eqref{wheightedbound} for $\lambda=\mu=1$. Hence, the functional dependence of $\Vert U \Vert_{q,p}$ on $(p,q)$ in the limit  $(p,q)\rightarrow(2,2)$ gives the optimal bound $c(1,1)$, in \eqref{wheightedbound}. 
For the case of the $\mathcal{L}^2(\mathbb{R})$-Fourier transformation  the norm $\Vert U_\mathcal{F}\Vert_{q,p}=\sqrt{p^{1/p}}/\sqrt{q^{1/q}}$ was computed by Beckner \cite{beckner}, leading to $c (1,1) = \log(\pi \rm e )$. 
However, to the best of our knowledge, computing $\Vert U\Vert_{q,p}$, for general $U$ and $(p,q)$, is an outstanding problem, and presumably very hard \cite{nphard,nphard2}. Albeit, for special choices of $(p,q)$ this problem gets treatable, see \cite{solvablecases} for a list of those. The known cases include $p=q=2$, $p=\infty$ or $q=\infty$ such as $p=1$ or $q=1$. 

The central idea of Maassen's and Uffink's work \cite{muff} is to show that the easy case of $(p=1,q=\infty)$, here we have $\Vert U\Vert_{1,\infty}=\max_{ij}|U_{ij}|$, gives a lower bound on $c(1,1)$. More precisely, they show that, for $1\leq p\leq 2$ and on the line $1/p+1/q=1$, the r.h.s. of \eqref{renyi4} approaches $c(1,1)$ from below. Note that this is far from being obvious. Explicitly, for $ p\leq 2 \leq q$ we have $H_{q/2}(Y|\phi) \geq H(Y|\phi)$ and $H_{p/2}(X|\phi)\leq H(X|\phi)$, so one term approaches the limit from above and the other approaches the limit from below. Whereas Maassen and Uffink showed, using the Riesz-Thorin interpolation \cite{riesz,thorin}, that the $\inf_\phi$ of the sum of both approaches the limit from below.   

The following Theorem, Thm.\ref{thm1}, extends the above to the case of arbitrary $(\lambda,\mu)$. Notably, we have to take $(p,q)$ from curves with $1/p+1/q \neq 1$, those are depicted in Fig.~\ref{fig:entrocurves}.  In contrast to Maassen and Uffink, the central inequality we use is the $\infty$-norm versions of the Golden Thompson inequality (see \cite{golden, thompson, goldenthompson} and the blog of T.Tao \cite{tao} for a proof and related discussions).   

\begin{thm}
\label{thm1}
Let $c(\lambda,\mu)$, with $\lambda,\mu \in \mathbb{R}_{+}$, be the optimal constant in the linear weighted entropic uncertainty relation 
\begin{align}
c(\lambda,\mu):=\inf_{\rho} \lambda H\left(X|\rho\right) + \mu H\left(Y|\rho\right). 
\end{align} Then:\vspace{0.1cm} 

\noindent(i) $c(\lambda,\mu)$ is bounded from below by $-N \log \left( \omega_N(\lambda,\mu)\right)$  
\begin{align}\label{rsstuff}
&\text{with \quad}\omega_N(\lambda,\mu)= \sup_{\substack{\x \in B_r(\mathbb{C}^d)\\ \y \in B_s(\mathbb{C}^d)}} \left| \x^\dagger U \y \right|
                \\&\text{and \quad} r=\frac{2N}{N+2\lambda}\quad s=\frac{2N}{N+2\mu}         
\end{align} 
where $$B_r(\Omega):=\left\{ \x \in \Omega | \; 1\geq\Vert \x \Vert_r \right\}$$
denotes the unit $r$-norm Ball on $\Omega$. \vspace{0.2cm}

\noindent(ii)
For $\lambda,\mu\leq N/2 $ we can write 
\begin{align}
&\omega_N(\lambda,\mu)= \sup_{\phi\in\mathbb{C}^d} \frac{\Vert U \phi\Vert_{r'}}{\Vert \phi \Vert_{s}}=\sup_{\phi\in\mathbb{C}^d} \frac{\Vert U \phi\Vert_{s'}}{\Vert \phi \Vert_{r}}
\\&\text{ with } r'=\frac{2N}{N-2\lambda}\quad s'=\frac{2N}{N-2\mu} \\            
\end{align}

\noindent(iii) For $\mu, \lambda \in \mathbb{R}^+\backslash\{0\}$, we have  
\begin{align}
c(\lambda,\mu) = \lim_{N \rightarrow\infty} -N \log \left( \omega_N(\lambda,\mu)\right)
\end{align} 
\end{thm}

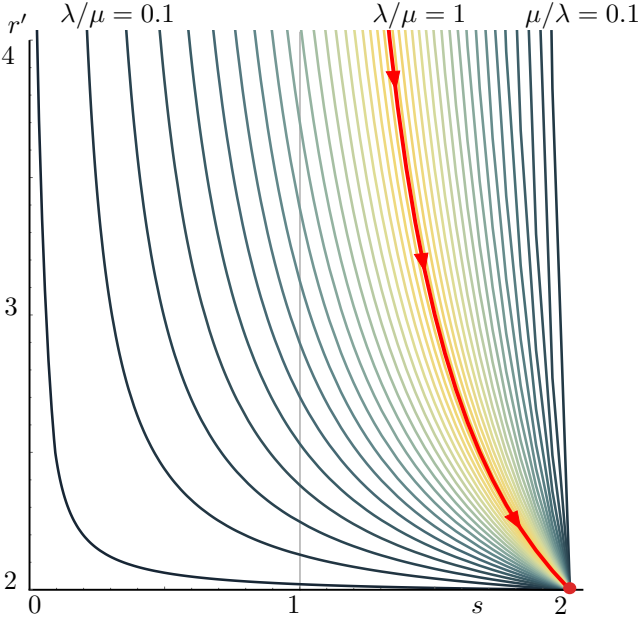
\begin{figure}[h]
\def\svgwidth{0.45\textwidth}
\input{entrocurves2.tex}
\caption{Evaluating $\Vert U \Vert_{r',s}$ on the depicted curves gives a lower bound for $c(\lambda,\mu)$, (see Thm.~\ref{thm1}). Because $c(\lambda,\mu)$ is a linear bound it is $1$-homogenious in $(\lambda,\mu)$. Hence all information on the optimal bound $c(\lambda,\mu)$ can be recovered by knowing it for any fixed ratio $\lambda/\mu$. The thick red curve corresponds to the case $1/{r'} +1/s=1$ which gives bounds $c(1,1)$ from below. For $s=1$ the norm $\Vert U \Vert_{r',s=1}$ can be computed analytically, this gives a generalization of the Masssen and Uffink bound (see Lem.~\ref{lem1}).
}
\label{fig:entrocurves}
\end{figure}

\begin{proof}
The starting point of this proof is a modification of a technique, used by Frank and Lieb in \cite{FrankLieb}, for reproving the Maassen and Uffink bound (see also the talk of Hans Maassen \cite{hanstalk}, for a finite dimensional version).

For probability distributions $\p,\q\in B_1(\mathbb{R}^d_+)$ we define the operators
\begin{align}
A(\p):=-\sum X_i \log(p_i)  \text{ and }  B(\q):=-\sum Y_i \log(q_i)
\end{align}
such that we can rewrite the Shannon entropy as 
\begin{align}
H\left(X|\rho\right)=\tr(\rho A(\px_\rho)) \text{ and } H\left(Y|\rho\right)=\tr(\rho B(\py_\rho))
\end{align}
Based on this, we can further rewrite the Shannon entropy as an optimization over a linear function in $\rho$ by using the positivity of the relative entropy, i.e. we have $D(\p||\q)=\sum p_i \log( p_i)- \sum p_i \log( q_i) \geq 0$, which implies 
$-\sum p_i \log( q_i)\geq H(\p)$. We obtain 
\begin{align}
H\left(X|\rho\right)=\inf_\p \tr(\rho A(\p)),
\end{align} such as the respective statement for $H\left(Y|\rho\right)$ and $B(\q)$. If we employ this rewriting to $c(\lambda,\mu)$, we obtain the minimal entropy sum as a  minimization over a parametrized eigenvalue problem, namely
\begin{align}
\label{minev}
c(\lambda,\mu)&=\inf_\rho \lambda H\left(X|\rho\right)+\mu H\left(Y|\rho\right) \nonumber\\&= \inf_{\p,\q,\rho} \tr\left(\rho \left(\lambda A(\p)+\mu B(\q)\right)\right )
\end{align}
Now we will turn the minimization, over $\rho$, into a maximization by applying the convex function $e^{-x/N}$, with $N\geq1$, to the weighted sum of $A$ and $B$. This will map the smallest eigenvalue of $\lambda A+ \mu B$ to the largest of $e^{-\frac{\lambda A(p)+\mu B(q)}{N}}$ and so on. In order to get back the correct value of $c$ we will have to apply the inverse function, $-N \log(x)$,  afterwards. We get
\begin{align}
c(\lambda,\mu)=-N \log\left( \sup_{\p,\q,\rho} \tr\left(\rho e^{-\frac{\lambda A(p)+\mu B(q)}{N}} \right) \right) .
\end{align}
Due to the positivity of the operator exponential, i.e. $A$ and $B$ are hermitian, the optimization over $\rho$ is equivalent to the Schatten-$\infty$ norm. We have
\begin{align}\label{eq:proof5}
c(\lambda,\mu)=-N \log\left( \sup_{\p,\q} \left\Vert e^{-\frac{\lambda A(p)+\mu B(q)}{N}} \right\Vert_\infty\right)
\end{align}
At this point we apply the Golden-Thompson inequality 
\begin{align}
\Vert e^{ S + T}\Vert_p\leq \Vert e^S e^T\Vert_p
\end{align}
and expand the resulting exponentials, as well as the Schatten norm. We get
\begin{align}\label{eq:proof4}
c(\lambda,\mu)&\geq  \text{-}N \log\left( \sup_{\p,\q} \left\Vert e^{-\frac{\lambda A(p)}{N}}e^{-\frac{\mu B(q)}{N}}  \right\Vert_\infty\right) \\
 &=     \text{-}N \log\left( 
                \sup_{\p,\q} \left\Vert \sum_{ij} X_i p_i^{\lambda/N} Y_j q_j^{\mu/N}  
             \right\Vert_\infty \right) 
\\
 &=    \text{-}N \log\left( 
                \sup_{\substack{\p,\q \\ \ket{x},\ket{y}}} \bigl\langle x \bigl| \sum_{ij} X_i p_i^{\lambda/N} Y_j q_j^{\mu/N}  
             \bigr|y \bigr\rangle \right)\label{blahaha}
\end{align}
Now we substitute $p_i^{\lambda/N}=:\chi_i$ and $q_j^{\mu/N}=:\xi_j$, and expand $\ket{x}=\sum x_i\ket{e_i}$ and $\ket{y}=\sum y_j\ket{f_j}$, with component vectors $\x,\y \in B_2(\mathbb{C}^d)$. By this the r.h.s of \eqref{blahaha} becomes
\begin{align}\label{eq:proof3}
 -N \log\left( 
                \sup_{\mathbf{\chi},\mathbf{x}} \sup_{\mathbf{\xi},\mathbf{y}} \left| \sum_{ij} \chi_i x_i \braket{e_i|f_j}  \xi_j y_j \right|
                \right).
\end{align} 
Here we can identify $\braket{e_i|f_j}=U_{ij}$, i.e. the overlaps are the components of $U$ when represented in the basis $X$. At this point, it is straightforward to check that  $\chi\in B_{N/\lambda}(\mathbb{R}^d_+)$ and $\xi\in B_{N/\mu}(\mathbb{R}^d_+)$. Using the generalized Hölder inequality we can fuse some of the maximizations above as follows:  On one hand, we have 
\begin{align}
&\qquad&\left(\sum | \chi_i x_i|^r \right)^{\frac1 r} &\leq \Vert \x \Vert_2 \Vert \chi \Vert_{N/ \lambda}\leq 1  \\ &\text{ and }&
\left(\sum | \xi_j y_j|^s \right)^{\frac1 s} &\leq \Vert \y \Vert_2 \Vert \xi \Vert_{N/ \mu}\leq 1 \nonumber\\
&\text{ for }& \frac1 r =\frac 1 2+\frac\lambda N &\quad\text{and}\quad \frac1 s =\frac 1 2+\frac\mu N \quad ,
\end{align}
which means that the vectors $\v$ and $\w$, with $v_i=\chi_i x_i$ and $w_j=\xi_j y_j$, are in $B_r(\mathbb{C})$ and $B_s(\mathbb{C})$ respectively.

On the other hand, the converse is also true, i.e. every $\v$ and $\w$ from $B_r(\mathbb{C}^d)$ and $B_s(\mathbb{C}^d)$ can be realized by suitable choices of $\x,\chi$ and $\y,\xi$. For example, we can always set 
\begin{align}
x_i=|v_i|^{r\lambda/N} \text{ and } \chi_i=|v_i|^{2/r} e^{i\operatorname{arg}(v_i)}
\end{align}
for getting $\x$ and $\chi$ from $\v$, componentwise. For this particular choice we can check that
\begin{align}
x_i\chi_i&=|v_i|^{r(\lambda/N+1/2)}e^{i\operatorname{arg}(v_i)}\nonumber\\&=|v_i|^{r/r}e^{i\operatorname{arg}(v_i)}=v_i
\end{align}
holds, such that we will get back $\v$. Furthermore $\x\in B_{N/\lambda}(\mathbb{R}_+^d)$ and $\chi\in B_{2}(\mathbb{C}^d)$ follows by writing out
\begin{align}
\sum_i x_i^{N/\lambda}=\sum_i v_i^r\leq 1 \text{ and }\sum_i \chi_i^2=\sum_i v_i^r\leq 1.
\end{align}
If we use the above in \eqref{eq:proof3}, we can replace $\sup_{x,\chi}$ by $\sup_\v$ and $\sup_{y,\xi}$ by $\sup_\w$, in order to
 get the statement (i) with 
\begin{align}\label{min2}
\omega_N:=\sup_{\substack{\v \in B_r(\mathbb{C}^d)\\ \w \in B_s(\mathbb{C}^d)}} \left|\v^\dagger U\w\right| \quad.
\end{align}
For showing the statement (ii), we take $r'$, with $1=1/r+1/r'$. If $\lambda \leq N/2$ holds we have $r'\geq 0$ and we can use the tightness of the Hölder inequality to rewrite 
\begin{align}
\sup_{\v\in B_r} \left|\v^\dagger U\mathbf{w} \right|=\Vert U \w\Vert_{r'}\; ,
\end{align}
i.e. the maximization over $B_r$ gives the dual norm of $r$. Substituting $\w$ by $\phi=\w\Vert\phi\Vert_s$ then gives 
\begin{align}
\omega_N=\sup_{\phi\in\mathbb{C}^d} \frac{\Vert U \phi\Vert_{r'}}{\Vert  \phi \Vert_s}
\end{align}
Here the analogous rewriting applies with $s'$ given by $1=1/s+1/{s'}$, if $\mu\leq \lambda/2$ holds.

For showing (iii), i.e.
\begin{align}
c=\lim_{N\rightarrow\infty} -N\log(\omega_N)\quad,
\end{align}
it suffices to expand all exponentials in \eqref{eq:proof5} and \eqref{eq:proof4} up to the first order in $N$. On this order the Golden-Thomson inequality is a equality. 
\end{proof}

\begin{remark}[The Maassen and Uffink bound]\label{muffremark}\normalfont
For the case of $N=2$ and $\lambda=\mu=1$, in Thm.\ref{thm1}, we get $s=r=1$ and $s'=r'=\infty$. Hence, we recover the Maassen-Uffink bound \cite{muff}. Explicitly, we have 
\begin{align}\label{muff}
\omega_2(1,1)= \sup_{\substack{\x \in B_1(\mathbb{C}^d)\\ \y \in B_1(\mathbb{C}^d)}} \left| \x^\dagger U \y \right|=\max_{ij}\left| U_{ij}\right|\;.
\end{align} Here we used that $\left| \x^\dagger U \y \right|$ is convex in $\x$ and $\y$. Hence, $\sup_{\x,\y}$ is attained at the extreme points of $ B_1(\mathbb{C}^d)$. Up to a phase, those extreme points have the form $(0,\cdots,0,1,0\cdots,0)$ , i.e. they have their support only on a single site. So, choosing $\x$ and $\y$, with support on the $i-th$ and $j-th$ site, will give $\left| \x^\dagger U \y \right|=\left|U_{ij}\right|$. 
\end{remark}

\begin{remark}[Renyi-Entropies]\normalfont
Alternatively, the bound obtained in Thm.~\ref{thm1} can be expressed in terms of Renyi-entropies: Using statement (i), (ii) and (iii) together  directly gives 
\begin{align} \label{renyi5}
 c(\lambda,\mu) &\geq -N\log\Vert U \Vert_{r',s} \nonumber\\&= \inf_{\phi \in \mathcal{H}} N \log \Vert \phi \Vert_r -N \log \Vert U \phi \Vert_{s'}.
\end{align}  Here a straightforward computation shows 
\begin{align}
\frac{2-r}{r}=\lambda/N \text{ and } \frac{2-s'}{s'} = -\mu/N  . 
\end{align}
So, when we proceed as in \eqref{renyi1}, substituting the Renyi entropy in \eqref{renyi5} gives
\begin{align} \label{renyi6}
 c(\lambda,\mu) \geq \inf_{\phi \in \mathcal{H}} \lambda H_{r/2} (X|\phi)+\mu H_{s'/2} (Y|\phi).
\end{align}
\end{remark}

\begin{lem}[Generalization of the Maassen and Uffink bound]
\label{lem1}
Let $\mathbf{u}_i$ denote the $i$-th column of the basis transformation $U$ that links the measurements $X$ and $Y$. Then, for $1\geq\lambda\geq\mu\geq 0$  and all states $\rho$ we have 
\begin{align}\label{generalmaassen}
\lambda H (X|\rho) + \mu H(Y|\rho) \geq -2\lambda\log \left( \sup_{i=1\cdots d} \bigl\Vert \mathbf{u}_i \bigr\Vert_t  \right) .
\end{align} 
with
\begin{align}
t=\frac{2}{(1-\mu/\lambda)}
\end{align}
Note that for the case $1\geq\mu\geq\lambda\geq 0$ the same holds, if $U$ is replaced by $U^\dagger$, i.e. by the transformation between $Y$ and $X$.
\end{lem}
\begin{proof}
The linear uncertainty bound $c(\lambda,\mu)$ is homogeneous in $(\lambda,\mu)$. Hence, we can consider 
\begin{align}
c(\lambda,\mu)=\lambda c(1,\mu/\lambda)
\end{align} We will apply Thm.~\ref{thm1}, with $N=2$, in order to get a lower bound. Here, we have  $s=\frac{2}{1+\mu/\lambda}$ and 
\begin{align}\label{generalmaassen1}
\omega_2(1, \mu/\lambda)=\sup_{\substack{\x \in B_1(\mathbb{C}^d)\\ \y \in B_s(\mathbb{C}^d)}} \left| \x^\dagger U \y \right|
						=\sup_{\substack{i = 1,\cdots, d\\ \y \in B_s(\mathbb{C}^d)}} \left| \mathbf{u}_i\;\y \right| .
\end{align} Here the second equality stems from the same argumentation as in Rem.~\ref{muffremark}. The $\sup$ over $B_s(\mathbb{C}^d)$ on the most right of \eqref{generalmaassen1}, gives the norm dual to $s$, given by $t=\frac{2}{1-\mu/\lambda}$. All in all we have, 
\begin{align}
c(1,\mu/\lambda)&\geq -2\log\left(\omega(1,\mu/\lambda)\right)\nonumber\\&= -2 \log \left( \sup_{i=1\cdots d} \bigl\Vert \mathbf{u}_i \bigr\Vert_t  \right) 
\end{align}
\end{proof}

\begin{remark}[More than two observables]\label{threeobs}\normalfont 
As mentioned in Sec.~\ref{mainres}, the proposition Prop.~\ref{prop1} does not generalize to three measurements. A reasoning, or at least a hint, for this can be found by carefully following the proof of Thm.~\ref{thm1}. In principle, the ansatz in \eqref{minev} can be generalized to more than two measurements as well, and all following steps work out in a similar way, up to \eqref{eq:proof4}. Here the Golden-Thompson inequality was used. It is well known, that the direct generalization of this inequality to three operators fails to hold. Hence, the technique of our proof cannot be generalized for this case. We note that there is an ongoing work of exploring more sophisticated generalizations of this inequality \cite{goldenimpro1,goldenimpro2,goldenimpro3,goldenimpro4}. However, we leave relating this to entropic uncertainty for future work.  
\end{remark}


%% file: entrocurves2.tex
\pdfoutput=1
\begingroup%
  \makeatletter%
  \providecommand\color[2][]{%
    \errmessage{(Inkscape) Color is used for the text in Inkscape, but the package 'color.sty' is not loaded}%
    \renewcommand\color[2][]{}%
  }%
  \providecommand\transparent[1]{%
    \errmessage{(Inkscape) Transparency is used (non-zero) for the text in Inkscape, but the package 'transparent.sty' is not loaded}%
    \renewcommand\transparent[1]{}%
  }%
  \providecommand\rotatebox[2]{#2}%
  \ifx\svgwidth\undefined%
    \setlength{\unitlength}{718bp}%
    \ifx\svgscale\undefined%
      \relax%
    \else%
      \setlength{\unitlength}{\unitlength * \real{\svgscale}}%
    \fi%
  \else%
    \setlength{\unitlength}{\svgwidth}%
  \fi%
  \global\let\svgwidth\undefined%
  \global\let\svgscale\undefined%
  \makeatother%
  \begin{picture}(1,1.00417827)%
    \put(0,0){\includegraphics[width=\unitlength]{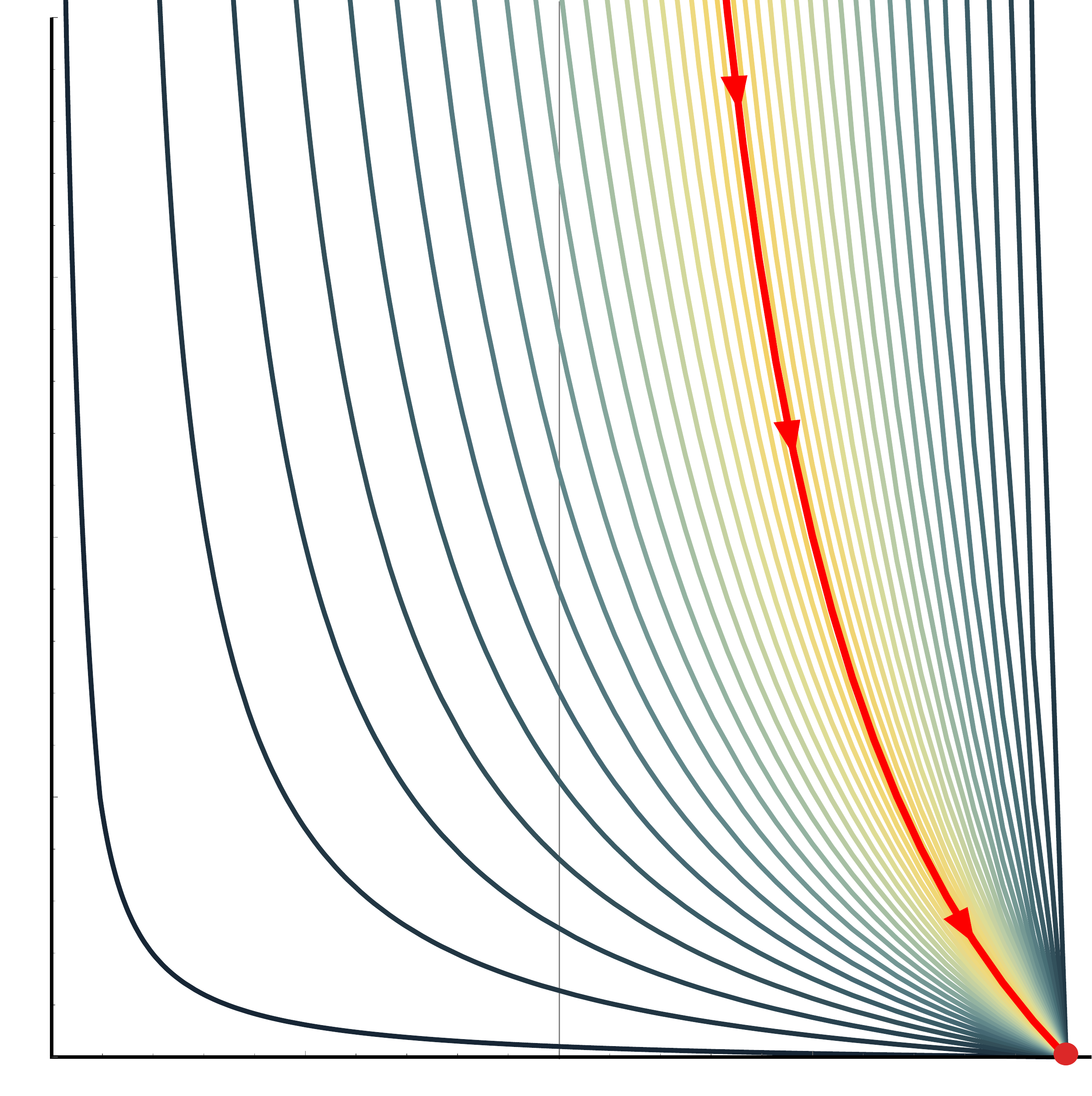}}%
    \put(-0.00075469,0.94717619){\color[rgb]{0,0,0}\makebox(0,0)[lb]{\smash{$4$}}}%
    \put(0.0040921,0.51138743){\color[rgb]{0,0,0}\makebox(0,0)[lb]{\smash{$3$}}}%
    \put(0.00300822,0.03679451){\color[rgb]{0,0,0}\makebox(0,0)[lb]{\smash{$2$}}}%
    \put(0.04526462,-0.00509331){\color[rgb]{0,0,0}\makebox(0,0)[lb]{\smash{$0$}}}%
    \put(0.49071858,-0.00511466){\color[rgb]{0,0,0}\makebox(0,0)[lb]{\smash{$1$}}}%
    \put(0.94945813,-0.00572608){\color[rgb]{0,0,0}\makebox(0,0)[lb]{\smash{$2$}}}%
    \put(0.01,1){\color[rgb]{0,0,0}\makebox(0,0)[lb]{$r'$}}%
    \put(0.80670047,-0.00537372){\color[rgb]{0,0,0}\makebox(0,0)[lb]{\smash{$s$}}}%
    \put(0.63699991,1.02){\color[rgb]{0,0,0}\makebox(0,0)[lb]{\smash{$\lambda / \mu =1$}}}
    \put(0.1,1.02){\color[rgb]{0,0,0}\makebox(0,0)[lb]{\smash{$\lambda / \mu =0.1$}}}
    \put(0.9,1.02){\color[rgb]{0,0,0}\makebox(0,0)[lb]{\smash{$\mu/\lambda=0.1$}}}%
  \end{picture}%
\endgroup%

%% file: additivity.tex
In this section we will provide the proof of Prop.~\ref{prop1}, i.e. the additivity of linear uncertainty relations.
By using Thm.\ref{thm1} from the section before we can formulate the linear uncertainty in terms of the logarithm of a $(p,q)$-norm. At this point, it is straightforward to check that the additivity of the linear uncertainty is equivalent to the multiplicativity of the  $(p,q)$-norm. In fact, the following theorem Thm.\ref{thm2} provides that, for $p$ and $q$ coming from the correct range: The $(p,q)$ norm of a transformation which admits a product form $U_{AB}=U_A\otimes U_B$ is multiplicative. 

\begin{thm}[Global bounds from local bounds]
\label{thm2}

Let $X_{AB}$ and $Y_{AB}$ be tensor-product bases of a Hilbert space $\mathcal{H}_{AB}=\mathcal{H}_A\otimes\mathcal{H}_B$, i.e. we have 
$X_{AB}=\{X_A^i\otimes X_B^i\}_{i=1,\cdots,d}$ and $Y_{AB}=\{Y_A^i\otimes Y_B^i\}_{i=1,\cdots,d}$, such as $U_{AB}=U_A\otimes U_{B}$. Furthermore let $\eta_A$ and $\eta_B$ denote the optimal constants for 
\begin{align}\label{locbounds}
\Vert U_A \phi\Vert_q\leq \eta_A\Vert \phi\Vert_p \quad \forall \phi\in\mathcal{H}_A\nonumber\\
\Vert U_B \phi\Vert_q\leq \eta_B\Vert \phi\Vert_p \quad \forall \phi\in\mathcal{H}_B \;.
\end{align}
If $1\leq p\leq q$ then
\begin{align}\label{globbound}
\Vert U_{AB}\; \phi\Vert_q\leq \eta_A \eta_B \Vert \phi\Vert_p \quad \forall \phi\in\mathcal{H}_{AB}
\end{align}
holds with $\eta_A \eta_B=\eta_{AB}$ as optimal constant.
\end{thm}
\begin{proof}
We note that a related result, for pointwise positive maps between Lebesque spaces, was discovered by Grey and Sinnamon \cite{simmons}.

The basic object of this proof will be the $p\otimes q$-norm which will be defined immediately. The basic work of this proof is devoted to show some properties of this norm from which the statement directly follows. 

Let $\ket{\phi}\in\mathcal{H}$ with components $\phi=\{\phi_{ij}\}$ sorted within the product base $X_{AB}$ by $\phi_{ij}=\braket{\phi|e^A_i\otimes e^B_j}$ and consider the norm
\begin{align}
\left\Vert \phi \right\Vert_{q\otimes p}:= \left( \sum_i \left( \sum_j \left| \phi_{ij} \right| ^p \right)^{\frac{q}{p}}     \right)^{\frac{1}{q}}\;.
\end{align}
This norm shares the following properties
\begin{align}
&(i)  &\left\Vert \phi \right\Vert_{q\otimes q}                     &=\left\Vert \phi \right\Vert_q\\
&(ii) &\left\Vert(\mathbb{I}\otimes V \phi) \right\Vert_{r\otimes q}&\leq\left\Vert \phi \right\Vert_{r\otimes p}\eta_V\\
&(iii)&\left\Vert \phi \right\Vert_{q\otimes p}                     &\leq\left\Vert \mathbb{F}\phi \right\Vert_{p\otimes q}\\
&\text{with } & \mathbb{F}\phi_1\otimes\phi_2=\phi_2\otimes\phi_1 &\text{ and } p\leq q \nonumber .
\end{align}
We will show the validity of $(i-iii)$ in a moment. First notice that, if $(i-iii)$ are valid we can easily conclude
\begin{align} \label{locfromglob}
\Vert U_{AB} \phi\Vert_q&=            \Vert U_A\otimes U_B \phi\Vert_{q\otimes q}                                \nonumber \\
				   &=            \Vert (\mathbb{I}\otimes U_B)(U_A \otimes \mathbb{I}) \phi\Vert_{q\otimes q} \nonumber \\
				   &\leq \eta_B     \Vert                     U_A \otimes \mathbb{I} \phi \Vert_{q\otimes p} \nonumber \\
				   &\leq \eta_B     \Vert \mathbb{I}\otimes U_A      \mathbb{F}      \phi \Vert_{p\otimes q} \nonumber \\
				   &\leq \eta_B \eta_A \Vert                          \mathbb{F}      \phi \Vert_{p\otimes p} \nonumber \\
				   &=    \eta_B \eta_A \Vert                          			    \phi \Vert_{p         }  \;.
\end{align} 
Furthermore, if we consider states that realize equality in \eqref{locbounds}, i.e. states that belong to optimal $\eta_A$ and $\eta_B$. The tensor-product of two of those states will realize, due to multiplicativity of the $p$-norm, equality in \eqref{globbound} as well. 
Hence, \eqref{locfromglob} will prove the main statement of this Theorem.

Property $(i)$  follows directly by plugging $p=q$ in the definition of the $p\otimes q$ norm, here is nothing more to prove.
The property $(ii)$ follows by expressing $\mathbb{I}\otimes V$ as $\delta_{ik}V_{jl}$ in $X$-Basis and 
\begin{align}
\left\Vert(\mathbb{I}\otimes V \phi) \right\Vert_{r\otimes q}&= \left( \sum_i \left( \sum_j \left|\sum_{lk} \delta_{ik} V_{jl} \phi_{kl} \right| ^q \right)^{\frac{r}{q}}     \right)^{\frac{1}{r}}\\
&= \left( \sum_i \left( \sum_j \left|\sum_{l} V_{jl} \phi_{il} \right| ^q \right)^{\frac{r}{q}}  \right)^{\frac{1}{r}}\\
&= \left( \sum_i \Vert V \phi_{i}\Vert_{q}^{\; r}  \right)^{\frac{1}{r}}\\
&\leq \eta_{V}\left( \sum_i \left( \sum_j \left| \phi_{ij} \right| ^p \right)^{\frac{r}{p}}  \right)^{\frac{1}{r}}\nonumber\\&= \eta_{V}\left\Vert \phi \right\Vert_{r\otimes p} \;.
\end{align}

As a last step, $(iii)$ is a direct consequence of Minkowski's inequality / $l^p$-triangle inequality (see \cite{HardyLittlewoodPolya}), i.e. if $p\geq 1$ : 
\begin{align}
\left(\sum_y \left|\sum_x a_{xy} \right|^p \ \right)^{\frac 1 p} \leq\sum_x\left(\sum_y\left|a_{xy}\right|^p \right)^{\frac{1}{p}} 
\end{align}
So, if $1\leq q/p $ we can use this inequality as follows  
\begin{align}
\left\Vert \phi \right\Vert_{q\otimes p}&= \left( \sum_i \left( \sum_j \left| \phi_{ij} \right|^p \right)^{\frac{q}{p}}\right)^{\frac{1}{q}} \nonumber\\ 
&=\left( \sum_i \left| \sum_j \left| \phi_{ij} \right|^p \right|^{\frac{q}{p}}\right)^{\frac{1}{q/p} \frac{1}{p}} \nonumber\\
&\leq\left( \sum_j \left( \sum_i \left| \phi_{ij} \right|^{p\frac{q}{p}}  \right)^{\frac{p}{q}}\right)^{\frac{1}{p}} 
=\left\Vert \mathbb{F}\phi \right\Vert_{p\otimes q}
\end{align}
and show the validity of $(iii)$. 
\end{proof}
\begin{lem}[Multiplicativity of the $(p,q)$-norm]\label{lem2}\quad\\ For $1\leq p \leq q$, the $(p,q)$-norm of a product unitary $U_{AB}=U_A\otimes U_B$ is multiplicative,  i.e. we have 
\begin{align}
||U_{AB}||_{q,p} = ||U_A||_{q,p} ||U_B||_{q,p}. 
\end{align}
\end{lem}
\begin{proof}
This is a direct consequence of Thm.~\ref{thm2}. Using the definition of the $(p,q)$-norm we can parse $\eta_A= ||U_A||_{q,p}$, $\eta_B= ||U_B||_{q,p}$ and $\eta_{AB} = ||U_{AB}||_{q,p} $, if we consider $\eta_A$, $\eta_B$ and $\eta_{AB}$ to be optimal bounds. 
\end{proof}

\textbf{Proof of Prop.~\ref{prop1}}\begin{proof}
For proving Prop.~\ref{prop1} it suffices to proof the additivity of the optimal case, i.e. we will consider $c_{A}$, $c_B$ and $c_{AB}$ to already be constants for the best linear uncertainty bound. If the additivity 
\begin{align}
c_{AB}=c_{A}+c_B
\end{align} holds we can directly conclude that the sum of lower bounds on $c_{A}$ and $c_B$ gives a valid lower bound on $c_{AB}$ as well. 

Given measurements $X_{AB}$ and $Y_{AB}$, specified by a product unitary $U_{AB}=U_A\otimes U_B$, we use Thm.~\ref{thm1} to rewrite $c_{A}$, $c_B$ and $c_{AB}$ as the limit of logarithms of $(p,q)$-norms. We assume $\lambda\leq\mu$ both to be finite and $N$ to be sufficiently large such that we can use Thm.~\ref{thm1} part (ii) (here we needed $\lambda,\mu\leq N/2$), and get
\begin{align}
c_A   &= -\lim_{N\rightarrow\infty} \; \log\left( ||U_A   ||_{r,s}\right) \\
c_B   &= -\lim_{N\rightarrow\infty} \; \log\left( ||U_B   ||_{r,s}\right) \\
c_{AB}& =-\lim_{N\rightarrow\infty} \; \log\left( ||U_{AB}||_{r,s}\right)
\end{align}  
Using $r,s$, as given in \eqref{rsstuff}  it is straightforward to check that $\lambda,\mu\leq N/2$ implies $1\leq r\leq s $. Therefore, we can use Lem.~\ref{lem2} and get 
\begin{align}
c_{AB}& =-\lim_{N\rightarrow\infty}\; \log \left( ||U_A   ||_{r,s} ||U_B  ||_{r,s}\right) \nonumber\\
      & =-\lim_{N\rightarrow\infty}\; \log \left( ||U_A   ||_{r,s}\right) + \log\left( ||U_B   ||_{r,s}\right) \nonumber\\
      & = c_A+c_B \;.
\end{align}

\end{proof}

%% file: outlook.tex
\pdfoutput=1

In this work we showed that linear uncertainty relations between product type measurements in multipartions are additive. 
Prop.~\ref{prop1} gives some clear structure to the problem of computing entropic uncertainty bounds. Especially in the context of quantum-coding in cryptography, this result might turn out to be useful, because now it is possible to compute uncertainty bounds in the limit of infinite system sizes for block-coding schemes \cite{ballesterwehner,winterlocking,furrer}.   

The generalization of the Maassen and Uffink bound for arbitrary weights $(\lambda,\mu)$, provided in Lem.~\ref{lem1}, can also be directly employed in a multipartite setting in order to obtain valid state-independent uncertainty relations for this case.
However, this bound is easy computable, it is only a lower bound and presumably only tight in high symmetrical cases (see \cite{ourentro} for a characterization of tightness for the usual Maassen and Uffink bound). The more general problem of providing a 'good' method for computing the optimal bound $c_{AB}$ remains open. We note that there are only few and special cases, including angular momentum and mutual unbiased bases, where this optimal bound is actually known. Thereby, the cases where the optimal bound can be computed analytically are even fewer \cite{sanchez2, sanchez, ourentro} and the known numerical methods only work for very small dimensional problems \cite{alberto}. Here the proof of Thm.~\ref{thm1} might give a new ansatz for better numerics. Explicitly, the minimization in \eqref{minev} and maximization in \eqref{min2} are giving rise to apply the method of alternating minimization.

In Sec.~\ref{pqnorms} we presented an extension to the known connections between the logarithm of $(p,q)$-norms and linear uncertainty relations in terms of the shannon entropy. However, the technique used seems to apply only for the special case we considered. An adaption of this technique to sets of more than two local measurements is not possible without major modifications. As mentioned in Rem.~\ref{threeobs}, this would require to incorporate generalizations of the Golden-Thompson inequality which seems to be a fruitful topic for future work.
The technique from the proof of Thm.~\ref{thm1} might also fail if general POVMs instead of projective measurements are considered. Moreover, it is  not clear if Prop.~\ref{prop1} will hold in this case. A third generalization, that does not hold, arises by considering arbitrary Schur-concave functions. Here, the natural question is to ask if at least any entanglement witness can be constructed. A very recent result \cite{timo} shows that such witnesses, in fact, can be constructed from Tsallis entropies.